\documentclass[10pt,journal,compsoc]{IEEEtran}
\usepackage{amsmath,amssymb,amsthm,graphicx,booktabs,array,multirow}
\usepackage{algorithm,algpseudocode}
\usepackage{newtxtext,newtxmath}
\usepackage{xcolor,url,cite}
\usepackage[hidelinks]{hyperref}
\graphicspath{{figures/}}
\newtheorem{theorem}{Theorem}
\newtheorem{corollary}{Corollary}
\newtheorem{proposition}{Proposition}
\newcommand{\us}{\,\mu\mathrm{s}}
\newcommand{\ms}{\,\mathrm{ms}}

\newcommand{\Prb}{\mathbb{P}}

\title{Finite-Calendar Admission and Deadline-Aware Pruning for UWB Audio}
\author{Fengyun~Zhang,~\IEEEmembership{Member, IEEE, and Member, AES}%
\thanks{Fengyun Zhang is with the College of Artificial Intelligence, Southwest University, Chongqing, China (e-mail: zhangfy2019@mail.sustech.edu.cn).}}
\begin{document}
\maketitle
\begin{abstract}
Isochronous ultra-wideband (UWB) audio transport must deliver a complete service data unit (SDU) before its presentation deadline; fragment reception rate and aggregate throughput do not directly characterize this objective. We formulate admission and online service over a periodic calendar of finite exchange opportunities, explicitly accounting for fragmented SDUs, bounded attempts, calendar alignment, and timing uncertainty. For equal-duration, freely assignable opportunities, an interval-demand condition characterizes finite-job reservation feasibility, and an earliest-deadline-first construction realizes a feasible reservation. A separate necessary capacity bound identifies loads that cannot sustain complete-SDU delivery even with no transmission errors. For online overload control, we study a deadline-aware policy that discards an SDU when its outstanding fragments exceed the remaining eligible opportunities. Across 3,456 replicated scheduling runs and 432 supplemental sensitivity runs, pruning raises eight-stream burst-loss timely completion from 45.441\% to 59.399\%, a paired difference of $13.958\pm2.080$ percentage points (pointwise 95\% interval). At four streams, round robin instead attains slightly higher completion, while heavy overload causes pronounced per-stream imbalance. A finite design enumeration and an evaluated multi-stream fragmentation envelope quantify the interaction among scheduling, payload overhead, and capacity. These findings distinguish provable admission conditions from online deadline salvage and expose the aggregate-throughput versus per-stream-service tradeoff in UWB audio scheduling.
\end{abstract}
\section{Introduction}
Low-latency wireless audio is fundamentally a deadline-constrained transport problem. A receiver presents audio in frames with a prescribed temporal relationship to source capture. For each frame, the useful transport event is the completion of its service data unit (SDU) before a presentation deadline, not the successful delivery of an arbitrary fragment. Fragment-level reliability can improve while the fraction of playable frames deteriorates: a transmission that cannot finish its remaining fragments by the deadline consumes channel service without contributing an on-time SDU. Consequently, mean throughput, packet reception ratio, and even success-conditioned latency are insufficient as stand-alone performance criteria. Bluetooth Low Energy (LE) Audio formalizes SDU intervals, lifetime controls, and presentation timing at different protocol layers~\cite{btll,isoal,bap}; comparable discipline is necessary when an application constructs an audio service over an independently configured ultra-wideband (UWB) link.

UWB is relevant to short-range audio networks because its time-structured transmissions and configurable data exchanges offer an opportunity to coordinate multiple periodic audio streams. Published UWB audio work and independent transceiver studies demonstrate the application context and device-specific communication characteristics~\cite{zhang2025,imfeld2023}. Yet a radio's physical rate or nominal exchange duration does not determine whether several fragmented streams can complete on time. A shared medium provides a discrete set of usable transmit--acknowledge opportunities, potentially clustered into short intervals and separated by long calendar gaps. A frame may be released between opportunities; different frames may have unequal relative deadlines; and individual fragments may consume repeated attempts. These effects couple fragmentation overhead, retry policy, temporal alignment, and admission in a manner that average-rate capacity calculations cannot describe.

The scheduling challenge has two complementary parts. \emph{Before service begins}, the system must decide whether a known workload can reserve a sufficient number of exchange opportunities in each frame's legal release-to-deadline window. This is a finite-calendar feasibility question, which depends on individual windows rather than utilization alone. \emph{During service}, loss and contention can create partially transmitted SDUs that no longer have enough remaining opportunities to finish. A scheduler that continues to serve such frames wastes capacity, whereas a scheduler that drops them aggressively can improve aggregate delivery but starve particular streams. An admission guarantee and an online discard criterion therefore answer different questions: one certifies an entire reservation under stated service assumptions, while the other detects a necessary condition for an individual frame's infeasibility.

The central problem considered here is as follows. A transmitter offers streams of periodic SDUs, with release times $r_j$, presentation-derived deadlines $d_j$, and fragment workloads $m_j$. A periodic opportunity calendar contains intervals of usable service with exchange duration $C$ and a maximum attempt count $K_j$ per fragment. Given a finite job set and timing uncertainty, we seek (i) a precise feasibility condition for reserving $m_jK_j$ exchanges without violating any deadline and (ii) an online rule that avoids spending future service on SDUs with no possible timely completion. The primary outcome is the complete-SDU timely-delivery fraction, interpreted together with the lowest per-stream completion rate, wasted opportunities, and success-conditioned delay. This problem statement does not assume that a particular controller exposes unrestricted cross-connection slot dispatch: the mathematical reservation model and the available SDK ownership rules are explicitly distinguished.

Existing scheduling and reliability literature supplies several foundations but leaves a specific integration question. Earliest-deadline-first (EDF) ordering and interval matching are established scheduling tools~\cite{liu1973,hall1935}; timely throughput and wireless delivery under unreliable service have a substantial history~\cite{hou2012}. LE Audio studies examine connected isochronous streams, retransmission persistence, and flushing~\cite{yan2024,yan2026}, while periodic wireless service work addresses resource reservation and latency~\cite{rtwt2024,fips2025}. The problem is therefore not to relabel EDF, TDMA, or frame expiry as a new protocol. Rather, it is to connect a complete-SDU presentation contract to a sparse \emph{finite exchange calendar}, derive the exact reservation implication under transparent assumptions, and quantify what online expiration achieves when an offered workload violates practical capacity constraints. This distinction has direct engineering value: it separates stream configurations that are admission-feasible from configurations for which a high aggregate completion percentage conceals unacceptable receiver-level service.

\noindent\textbf{Contributions.} The contributions are threefold: \emph{(1) Finite-calendar admission analysis:} we map presentation deadlines into exchange eligibility intervals, establish an exact interval-demand criterion for finite workloads with equally sized, freely assignable opportunities, provide an EDF reservation construction, and derive aligned-batch and aggregate-capacity bounds; \emph{(2) Deadline-aware online service:} we characterize an optimistic impossibility-based SDU pruning rule and evaluate its tradeoff against five baseline disciplines under independent and burst-correlated exchange failures, including paired uncertainty estimates, ablations, and per-stream fairness analysis; \emph{(3) Protocol-aware evaluation:} we quantify how a stream- and epoch-aware fragmentation envelope changes feasible workload, enumerate finite transport configurations, and trace the model to executable scheduling code and the documented TDMA controller interface. Across these elements, analytical guarantees and numerical policy outcomes are reported as distinct results.

\noindent\textbf{Paper organization.} Section~II reviews isochronous transport and related wireless scheduling. Section~III describes the SDU adaptation and controller interface, and Section~IV formalizes the traffic and opportunity model. Section~V derives finite-calendar feasibility and capacity bounds. Section~VI specifies the online policy and bounded design search. Section~VII presents the numerical results and reproducibility checks. Section~VIII examines the controller mapping and timing semantics, Section~IX discusses the scheduling tradeoffs, and Section~X concludes the paper.

\section{Background and Related Work}
\subsection{Scope of Prior Work}
Four research areas intersect in this problem: UWB link implementation and audio transport; isochronous service and flushing semantics; deadline-aware scheduling on unreliable shared channels; and controller-calibrated evaluation. They overlap in vocabulary but differ in what is guaranteed. In particular, a throughput model does not establish a bound on complete-SDU age, and a controller lifetime limit does not by itself establish a capture-to-playback bound.

We examine the directly relevant UWB audio and transceiver studies, UWB interoperability literature, recent LE Audio modeling, multiparty communication, periodic wireless scheduling, and end-to-end reliability analysis~\cite{zhang2025,imfeld2023,coppens2022,yan2024,yan2026,mantis2026,rtwt2024,fips2025}. The comparison focuses on service semantics, controllable scheduling variables, and actual evidence rather than claiming priority over entire classes of work. The 2024 CIS performance model and 2026 lifetime/flushing analysis are particularly close conceptual comparators~\cite{yan2024,yan2026}: explicit packet lifetime and retransmission budgeting cannot be presented as new ideas merely because they are applied to another PHY. Our narrower analytical object is the finite, possibly gapped \emph{exchange calendar} and complete-SDU attempt workload; its implementation and deployment qualifications remain essential.

\subsection{Standards: What Can and Cannot Be Mapped}
The initial premise that IEEE 802.15.4z does not supply the complete LE Audio service model should not be inflated into a claim that the whole IEEE MAC lacks timing mechanisms. The 4z amendment is identified with enhanced UWB PHYs and associated ranging techniques~\cite{ieee4z}; this study has not exhaustively audited every normative MAC clause. UWB interoperability literature likewise cautions that commercial MAC and upper layers are frequently proprietary~\cite{coppens2022}. The SDK discussed below is a distinct, version-dependent implementation substrate; its use is neither evidence of IEEE 802.15.4z/4ab compliance nor a substitute for a standard-conformance assessment.

At the time of writing, the publicly accessible IEEE P802.15.4ab project page still describes an active draft project and names low-power, low-latency streaming and high-rate streaming among its scope items~\cite{ieee4ab}. Such scope language does not establish the content of nonpublic draft clauses or provide a ready-made deadline API for the platform studied here. The proposed SDU abstraction is a research-level scheduling interface rather than a normative amendment or a claimed standards implementation.

In the inspected Bluetooth Core 5.4 Link Layer, \texttt{ISO\_Interval} is expressed in 1.25 ms units, with values 4--3200, or 5 ms--4 s; CIS \texttt{NSE} is 1--31, and directional \texttt{BN} is 0--15~\cite{btll}. These are protocol parameters with specific packing and timing constraints, not arbitrary UWB design ranges. CIS acknowledgment and flushing should also not be copied indiscriminately onto BIS. Broadcast retransmission structure includes \texttt{IRC} and \texttt{PTO}; it is not the same feedback ARQ as a point-to-point CIS.

Most importantly, $\mathrm{FT}\,I$ is not a universal complete transport-latency expression. For framed CIS in the inspected Core 6.0 ISOAL, Section 3.2.1 gives
\begin{equation}\label{eq:bleframed}
 L_{\rm CIG}^{\rm framed}=G+\mathrm{FT}\,I+F_{\rm framing},
\end{equation}
where $G$ is CIG\_Sync\_Delay and the framing term depends on the mode. Section 3.2.2 gives, for unframed CIS,
\begin{equation}\label{eq:bleunframed}
 L_{\rm CIG}^{\rm unframed}=G+(\mathrm{FT}-1)I+(U-1)S,
\end{equation}
where $U$ is USPI and $S$ is SDU\_Interval~\cite{isoal}. The specification excludes implementation-specific processing/internal transport delays from this latency. For reception by a Unicast Server or Broadcast Sink, BAP Presentation Delay spans the SDU synchronization reference to presentation. For transmission by a Unicast Server, BAP Section 7.1.2 instead defines acquisition-start to SDU synchronization reference. Neither should be indiscriminately called microphone-to-ear delay~\cite{bap}. These distinctions motivate a fresh end-to-end budget instead of renaming FT.

\subsection{UWB Data and Audio Research}
Coppens \emph{et al.} organize standards, radios, and interoperability constraints~\cite{coppens2022}. Their distinction between PHY compatibility and MAC/upper-layer compatibility is particularly relevant: a common spectral designation does not imply that audio groups or time references interoperate. Imfeld \emph{et al.} measure SR1000 energy, data transfer, and ranging on a documented evaluation kit~\cite{imfeld2023}. That work supports the choice of a SPARK platform for investigation, but does not validate the present SR1100-family scheduler or supply a calibrated energy model for it.

Zhang's 2025 preprint is direct prior art, not merely a generic UWB prospect paper~\cite{zhang2025}. Its Section III describes a UWB audio framework, audio/wireless cores, synchronization and acknowledgment frames, bandwidth scheduling, and development-board audio evaluation. Consequently, a broad ``first UWB audio architecture'' claim is untenable on the retrieved evidence. Our work instead examines explicit SDU expiration and service reservations with inspectable source and provenance. The preprint's reported utilization and success figures are not imported as our results or treated as a matched baseline, because workloads, measurement definitions, and radios differ.

Ranging-focused interference evaluation also exists in recent work~\cite{tiemann2022}; it is not an audio transport benchmark. Nelson \emph{et al.}'s interference study combines measured spectral activity and a modeled UWB receiver~\cite{nelson2024}. We cite its explicitly fixed 2024 version, rather than conflate later title revisions. It warns against treating independent packet errors as a measured radio model. Historical high-rate WPAN scheduling also addresses reserved isochronous traffic~\cite{kim2014}; the fact that the present target is impulse-radio audio does not make periodic reservations conceptually unprecedented. This historical evidence and the SDK audit together move the research gap from ``add a superframe'' to ``define and validate the right SDU service contract.''

\subsection{LE Audio, Concealment, and Deadline Scheduling}
LE Audio research provides direct comparators. Yan \emph{et al.} study CIS performance and the influence of lifetime- and flushing-based mechanisms~\cite{yan2024,yan2026}; their work precludes attributing novelty to expiration policies alone. MANTIS studies implemented multiparty LE Audio communication using coordinated broadcast roles~\cite{mantis2026}, whereas the present test firmware is limited to a two-node point-to-point configuration. WiLE-Audio considers cross-technology extension of audio coverage~\cite{wile2025}, addressing a different physical-layer objective. These studies underscore why a fair cross-technology comparison needs a complete implementation and reconciled presentation time references. Our evaluation instead compares schedulers over \emph{one fixed modeled UWB opportunity service}.

Assistive-listening work and the Bluetooth audio evolution discussion emphasize deployment and application context~\cite{bruce2025,baghel2023}. LC3 characterization and BAP distinguish codec operating points from transport framing~\cite{lc3white,bap}. BAP lists nominal 7.5/10 ms modes; its 44.1 kHz note gives effective durations of approximately 8.16327/10.88435 ms. Therefore, our generic 15/20 ms batching sensitivity is not called a native LC3 frame mode. No decoder or packet-loss concealment (PLC) is executed here, and a missing-frame run is not a perceptual quality score.

Classical EDF and Hall matching provide the appropriate theoretical foundation~\cite{liu1973,hall1935}. Hou and Kumar explicitly distinguish timely throughput over unreliable wireless channels from conventional queue stability~\cite{hou2012}. Recent R-TWT analysis jointly selects periodic service duration and period under retransmission and delay constraints~\cite{rtwt2024}; FIPS stresses that wireless-segment reliability need not survive end-to-end scheduling composition~\cite{fips2025}. Pang \emph{et al.} combine BLE throughput/reliability modeling with practical validation~\cite{pang2024}. These works preclude claiming that deadline budgets, periodic access, or retry/reliability tradeoffs are themselves new. Table~\ref{tab:related} distinguishes the modeled contribution from related protocol and scheduling mechanisms.

\begin{table*}[t]
\caption{Approach landscape and evidence boundaries. ``Not established'' does not mean absent from an inaccessible paper.}\label{tab:related}
\centering\small
\begin{tabular}{p{2.9cm}p{4.0cm}p{4.3cm}p{4.3cm}}\toprule
Work / substrate & Relevant scope & Evidence inspected & Difference relative to this study\\\midrule
UWB standards survey~\cite{coppens2022} & Radio and layer interoperability & Published full text & No audio QoS theorem inferred from radio compatibility.\\
SPARK characterization~\cite{imfeld2023} & SR1000 transfer, energy, ranging & Published metadata and author preprint & Different family/settings; does not calibrate our SR11xx service time.\\
UWB audio architecture~\cite{zhang2025} & Audio/WPS design and board experiment & Preprint Section III & Direct architecture prior art; we study explicit SDU deadlines, not first audio transport.\\
CIS and lifetime/flushing models~\cite{yan2024,yan2026} & Isochronous persistence/performance & Metadata and available abstract & Direct retry/expiry overlap; overlap concerns SDU lifetime and retries.\\
MANTIS / WiLE-Audio~\cite{mantis2026,wile2025} & Multiparty BIS / cross-technology coverage & MANTIS full text; WiLE metadata & System-level implementations; distinct scheduling objectives and evidence.\\
R-TWT / FIPS~\cite{rtwt2024,fips2025} & Periodic wireless service / end-to-end QoS & Accessible primary texts & Related budget composition and service calendars, different radio/application.\\
SPARK SDK 2.4.0~\cite{spark} & TDMA, bounded ARQ, fragmentation, PCM & Supplied headers, examples, documentation & Mechanisms inherited; application SDU semantics are the intended extension.\\
Present study & Calendar model, finite test, pruning policies & Analysis, code, seed-level tables & Finite calendar analysis and numerical outcomes; firmware interface independently audited.\\\bottomrule
\end{tabular}
\end{table*}

\section{Isochronous Adaptation and Controller Mapping}
\subsection{Application-Level Isochronous Adaptation}
Figure~\ref{fig:architecture} distinguishes the application scheduling contract, the vendor link implementation, and the receiver presentation path. The proposed adaptation boundary accepts a complete SDU together with its stream identity, release reference, deadline, and legal codec configuration. It supplies sequence-aware fragment PDUs and must report either a complete SDU or a frame-loss event before presentation. Native fragmentation can be reused if its cancellation and reassembly semantics are compatible; duplicating it without a reason would add overhead rather than a contribution.

The original byte-level reference and conditional firmware probe share an eight-byte envelope: a 32-bit little-endian sequence number, a 16-bit local-millisecond release field, and 8-bit fragment index/count. Its 240-byte known-pattern SDU occupies 96-, 96-, and 72-byte PDUs at a 96-byte PDU limit. The Python receiver supports fragment reordering, duplicate checking, bounded reassembly, and caller-supplied local deadlines. These properties establish single-context byte-level behavior only. In particular, the short timestamp is a local diagnostic field, not a negotiated global playback reference.

To distinguish single-context and multi-stream packet semantics, we additionally implement a \emph{distinct Python-only group envelope}. Its 24-byte header contains a group ID (16 bits), stream ID (16 bits), session epoch (32 bits), SDU sequence (32 bits), group-clock presentation deadline (64 bits), fragment index (16 bits), and fragment count (16 bits). Reassembly is keyed by $(\mathrm{group},\mathrm{stream},\mathrm{epoch},\mathrm{sequence})$ and enforces bounded state, consistent fragment metadata, and expiration. Interleaved streams with identical sequence numbers are independently reassembled in regression tests; epoch mismatch, active-SDU expiry, and bounded duplicate-suppression behavior are also tested. Completely absent SDUs require an external expected-release schedule to trigger loss indications. The numerical runs use the eight-byte format; the 24-byte Python envelope is a separate reassembly implementation and is not described as an interoperable security protocol. The timestamp is meaningful only after external group-clock calibration and authenticated epoch establishment.

The two reference formats should not be conflated. The original eight-byte receiver aliases distinct streams at equal sequence numbers unless external connection contexts isolate them; its 16-bit millisecond field wraps every 65.536 s. The 24-byte alternative repairs the immediate identity ambiguity but adds sixteen bytes per fragment, increases the modeled 240-byte SDU from three to four fragments at $P=96$ bytes, and therefore changes capacity. Cryptographic integrity, replay resistance beyond the bounded tombstone cache, clock recovery, cross-device deadline interpretation, congestion feedback, codec negotiation, and radio/firmware portability remain outside the Python extension. The two envelopes consequently have distinct workload and capacity parameters. Table~\ref{tab:envelope} makes this otherwise easily obscured payload cost explicit.

\begin{table}[t]
\caption{Two receiver envelopes (240-byte SDU, $P=96$ bytes); $(g,i,e,n)$ denotes group, stream, epoch, sequence.}
\label{tab:envelope}\centering\small
\begin{tabular}{p{2.4cm}p{2.0cm}p{2.7cm}}\toprule
Property & Probe-compatible & Multi-stream Python\\\midrule
Header & 8 B & 24 B\\
Body per PDU & 88 B & 72 B\\
Fragments/SDU & 3 & 4\\
Reassembly key & Sequence & $(g,i,e,n)$\\
Deadline & External local deadline & Group-clock deadline*\\
Radio integration & Compiled probe & Not integrated\\ \bottomrule
\multicolumn{3}{p{7.4cm}}{\footnotesize *The timestamp field is not evidence of calibrated cross-device clocks.}
\end{tabular}
\end{table}

\begin{figure*}[t]\centering
\includegraphics[width=.90\textwidth]{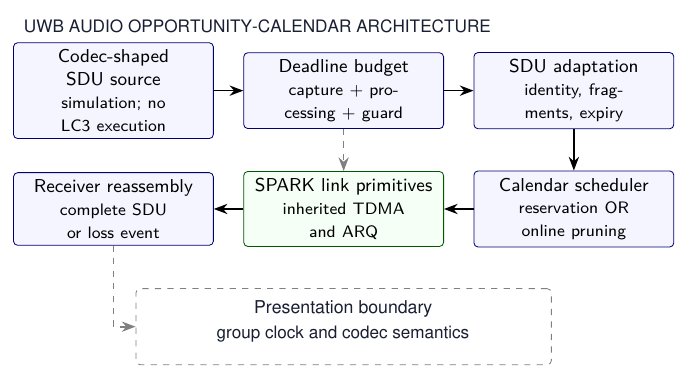}
\caption{Scheduling and adaptation boundaries. The finite admission test and online pruning policy are distinct; configured SDK connection/slot ownership sits below them. LC3 presentation and cross-device clock calibration are not part of the evaluated path.}\label{fig:architecture}
\end{figure*}

\subsection{Group Calendar, Control Plane, and Playout Boundary}
A conceptual UWB group shares a period, epoch, opportunity calendar, and presentation contract. Its control plane should negotiate legal PDU budgets, maximum SDU size, retry caps, timing uncertainty, and the handling of a missed or incomplete SDU. The present model only multiplexes logical streams at one transmitter; it does not schedule arbitrary independent transmitters in a shared slot. The SDK's same-slot connection priorities have compatibility requirements and cannot be assumed to solve general distributed EDF~\cite{spark}. More fundamentally, the SDK's network-wide TDMA timeslot schedule is fixed while a network runs, even though configured connections may share a timeslot under a priority rule~\cite{spark}. The simulation's ability to select any eligible stream in any modeled opportunity is therefore an \emph{upper-level allocation relaxation}. An implementable variant must either multiplex several application streams inside one legally configured connection (and select complete SDUs before queue insertion) or restrict choices to compatible connections assigned to a shared timeslot. In both cases, controller FIFO behavior, internal retries, and in-flight packets can obstruct atomic expiration. The independently compiled probe does not establish control of these operations.

The control plane should check per-connection eligibility and conservative admission before enabling another stream. If admission succeeds and measured exchange service is bounded as modeled, a reserved schedule can provide the promised attempt opportunities; it cannot promise RF success under arbitrary losses. Runtime pruning is a different, less conservative mode for overbooking or overload experiments. The implementation must advertise which mode is used; otherwise a system could silently replace a proved reservation with a heuristic and retain an invalid guarantee. Grouped receiver presentation also needs its own error bound, such as a maximum inter-receiver clock residual, separate from the single-link data deadline.

At the receiver, SDU reassembly and codec presentation must remain separate stages. A fragment arriving after a calibrated completion deadline should not delay later SDUs; a partial frame should trigger the codec's supported loss indication. Which concealment is best depends on actual decoded content and prior loss history. Therefore, the current architecture exposes a frame-loss event but does not claim to optimize perceptual concealment. BAP/ISOAL are useful design references for this separation, not evidence that the SPARK probe implements their timing semantics~\cite{isoal,bap}.

\section{Traffic, Deadline, and Opportunity Model}
\subsection{Timeline, Workload, and Clock Contract}
All modeled times are integer microseconds. A stream $i$ has capture-frame duration $F_i$, encoded bitrate $b_i$ bit/s, and SDU size
\begin{equation}\label{eq:size}
 B_i=\left\lceil\frac{b_iF_i}{8\cdot10^6}\right\rceil,\qquad
 m_i=\left\lceil\frac{B_i}{P-h}\right\rceil,
\end{equation}
where $P$ is the application PDU budget including an $h$-byte adaptation header. This constant-rate source emulates packet size and timing only. A real codec can impose different legal sizes, sample counts, look-ahead, or processing costs.

Let $r_j$ denote the time SDU $j$ becomes ready for radio service, $P_i^{\rm pres}$ the capture-to-presentation budget, and $L_i^{\rm other}$ the non-capture codec/IO/playout cost. Here $r_j$ denotes availability of the encoded SDU at the scheduling boundary; the subtraction of $F_i$ and $L_i^{\rm other}$ assumes the stated capture and non-radio stages are measured with the same reference convention. Encoding delay must be accounted for exactly once, not simultaneously deducted in this budget and counted again as an additional delay after $r_j$. The available radio budget is
\begin{align}\label{eq:budget}
 D_i &=P_i^{\rm pres}-F_i-L_i^{\rm other},\\
 d_j&=r_j+D_i-\epsilon_i,\label{eq:deadline}
\end{align}
where $\epsilon_i$ bounds the application clock/dispatch uncertainty. We require $D_i>\epsilon_i$. Baseline values $25\ms-10\ms-2\ms$ yield $D_i=13\ms$ and an effective $12.960\ms$ readiness-to-completion budget. Neither the 2 ms cost nor the 40 $\mu$s uncertainty is measured hardware timing.

At a real receiver, a group-clock deadline $d_j$ must be mapped onto its local oscillator. If $t_{\rm rx}=\gamma t_{\rm group}+\delta+e(t)$, the effective guard must cover uncertainty in $\gamma$, $\delta$, quantization, synchronization update age, and scheduling callback jitter. No such calibrated bound is available from the current probe. A short source timestamp does not provide a common clock. Accordingly, the simulator assumes the budget clock, the Python reassembler accepts an externally supplied local deadline, and the firmware records local times without claiming one-way latency.

\subsection{Sparse Opportunity Calendar and Exchange Cost}
A single central transmitter multiplexes logical streams over a half-duplex radio. It has a superframe $H$, exchange-opportunity duration $C$, nominal share $\rho$, and $Q=\lfloor\rho H/C\rfloor$ allocated opportunities. The executed uniform-prefix calendar is
\begin{equation}\label{eq:calendar}
 s_{\ell,q}=\ell H+qC,\quad q=0,\ldots,Q-1.
\end{equation}
All policies gate a ready SDU to $a_j=H\lceil r_j/H\rceil$. Thus, arrival alignment is identical in every comparison. The clustered calendar and anchor gating are research choices, not a requirement of all SPARK applications. Other transmissions, ranging, beacons, or guard intervals must be removed from the usable calendar before applying the analysis.

One opportunity contains at most one modeled DATA/ACK exchange, including its turn-around and service margin. The audited SR1120 IOOK/FEC1.25 documentation gives a DATA-only RF expression
\begin{equation}\label{eq:air}
 T_{\rm data}(N)=15.97+0.4883N\quad\mu\mathrm{s}.
\end{equation}
We do not interpret this as the complete exchange or minimum schedulable slot. The default $C=200\us$ is anchored to an SDK audio example; the validation check adds assumed ACK RF time 12 $\mu$s, turn-around 3 $\mu$s, and the chosen guard to the DATA expression before comparing with $C$. Once a full $C$ is charged, those components must not be charged again at completion.

The radio-on proxy adds DATA RF time and 12 $\mu$s per attempted exchange, even for a failed exchange. It is an intentionally approximate accounting measure, not a current trace. The primary resource metric is occupied slot time divided by the source-plus-drain horizon. CCA, sleep/wake cost, clock recovery, and RF adaptation are not calibrated. Actual throughput is determined by usable exchange opportunities and payload, not by equating UWB spectral bandwidth with payload bitrate.

\subsection{Losses, Service Work, and Objectives}
Each fragment permits $K_i$ \emph{total attempts}, including the first. A successful modeled exchange advances one fragment; exhausting its cap aborts the entire SDU, even if earlier fragments have already been received. The simulator abstracts DATA and ACK failure into one exchange-loss event. It does not model the asymmetry in which a receiver obtains DATA but the sender misses ACK. Hence, exchange-confirmed modeled completion must not be confused with a measured receiver-completion counter under ACK loss.

For independent exchange losses with probability $p$, an uncongested SDU with $m$ fragments has
\begin{equation}\label{eq:iid}
 R_{\rm iid}(m,K,p)=(1-p^K)^m.
\end{equation}
The formula requires enough timely attempts and suitable independence; it is not an overloaded-queue guarantee. A two-state burst process is evaluated separately. With good-state loss zero, bad-state loss one, $\beta=1/8$, and $\alpha=p\beta/(1-p)$, its stationary bad-state fraction is $p$. A transition occurs once per calendar opportunity, including idle opportunities, and the initial state is stationary. Thus, the average burst duration is eight \emph{opportunity indices}, not a fixed wall-clock coherence time. Changing $H$ or $C$ can consequently change the implied physical-time correlation structure. This deliberately simple member of the Gilbert--Elliott family~\cite{gilbert1960,elliott1963} is not fitted to a SPARK channel.

The empirical objective is timely complete-SDU fraction, accompanied by missing-frame runs, per-stream completion, conditional latency, and occupied opportunities. A design objective selects $(F,H,P,K,b)$ subject to service feasibility and target reliability. A perceptual distortion objective or PLC-aware utility would be a useful extension, but no such objective is evaluated. Table~\ref{tab:params} separates source-derived settings from explicit assumptions.

\begin{table}[t]
\caption{Main simulation parameters and provenance.}\label{tab:params}\centering\small
\begin{tabular}{p{2.6cm}p{1.7cm}p{3.1cm}}\toprule
Parameter & Baseline & Basis\\\midrule
Capture frame $F$ & 10 ms & Codec-shaped source assumption\\
Rate $b$ per stream & 192 kbit/s & Constant-byte workload, not LC3 execution\\
Presentation budget & 25 ms & Research choice\\
Other processing & 2 ms & Assumption, no codec timing\\
Clock/dispatch guard & 40 $\mu$s & Assumption\\
Superframe $H$ & 5 ms & Research choice\\
Exchange slot $C$ & 200 $\mu$s & SDK example, not measured minimum\\
Nominal share $\rho$ & 0.32 & Research choice; integer quantization\\
PDU/header $P,h$ & 96 / 8 B & Adaptation design\\
Fragments $m$ & 3 & $240/(96-8)$ rounded up\\
Total tries $K$ & 3 & Research choice; 2 retries\\
Exchange loss $p$ & 0.1 & Specified IID or burst chain\\
Runs & 12 seeds, 6 s & Seed-level independent replications\\\bottomrule
\end{tabular}
\end{table}

\section{Finite-Calendar Admission Analysis}
\subsection{Exact Finite-Calendar Reservation Condition}
Let $s_1<\cdots<s_M$ be nonoverlapping exchange opportunities of equal duration $C$. Each finite job $j$ has eligibility $a_j$, completion deadline $d_j$, and integer reservation work $w_j=m_jK_j$. Define its eligible slot-index set
\begin{equation}\label{eq:window}
 I_j=\{k:a_j\le s_k,\ s_k+C\le d_j\}.
\end{equation}
These sets are contiguous in slot-index space even when the wall-clock calendar contains gaps. Assumptions are explicit: one exchange fits one slot; service can be interleaved between exchanges; ACK outcome is available by the next usable exchange; there is no additional fragment cooldown; and there are no external precedence or shared-resource constraints beyond serialization. Reserved quanta are fungible attempt capacity, not simultaneous transmissions.

\begin{theorem}[Finite unit-exchange reservation]\label{thm:hall}
The finite jobs admit a reservation of $w_j$ slots per job inside $I_j$ if and only if each positive-work job has nonempty $I_j$ and, for every contiguous slot-index interval $[u,v]$,
\begin{equation}\label{eq:hall}
 \sum_{j:I_j\subseteq[u,v]}w_j\le v-u+1.
\end{equation}
Visiting slots chronologically and assigning each to an eligible unfinished job with earliest deadline returns a feasible reservation whenever one exists.
\end{theorem}
\begin{proof}
Necessity follows because jobs whose complete eligibility intervals are contained in $[u,v]$ cannot receive service elsewhere. For sufficiency, create $w_j$ unit clones of each job, each adjacent to every slot in $I_j$. A matching covering all clones is precisely a reservation. Suppose Hall's condition fails for some clone subset. Enlarge it to include all clones of every job it intersects; this cannot remove a deficiency, because added clones of the same jobs add no neighbors. The union of the corresponding interval neighborhoods decomposes into disjoint contiguous components. Each selected job interval belongs entirely to one component. Since the total selected demand exceeds the total neighborhood size, at least one component has selected demand exceeding its slot count. Counting all jobs whose intervals are contained in that component only increases the demand, contradicting~\eqref{eq:hall}. Hall's theorem therefore gives a covering matching~\cite{hall1935}.

For the EDF construction, take a feasible unit schedule agreeing with EDF up to the first different assignment. EDF's chosen clone $e$ is eligible now; in the feasible schedule it is assigned to a later slot. If the current slot is idle, move $e$ earlier. Otherwise the current slot serves a clone $x$ that is already released, with deadline no earlier than $e$'s. Swap their slots. Clone $x$ remains released at the later slot and meets its deadline because $e$ did; $e$ is eligible at the current slot by construction. The schedule now agrees with EDF for one more slot. Iteration establishes feasibility. This is a unit-calendar exchange argument, not a new unrestricted EDF optimality theorem~\cite{liu1973}.
\end{proof}

The condition concerns a finite, known set of jobs, and it relies on a complete bipartite eligibility relation between each job and all slots in its interval. Fixed connection ownership restricts that relation; a general SDK-owned assignment must use its actual job-slot graph rather than blindly apply the interval-only EDF construction. An aggregate utilization bound alone cannot replace the finite test. For example, let $H=5$ ms, $C=0.2$ ms and $Q=8$, so opportunities occupy only the first 1.6 ms of each period. A one-quantum job ready at 2 ms with deadline 3 ms has no eligible slot, despite its very small average demand. Anchor gating delays its eligibility further. This is exactly the type of failure that~\eqref{eq:window} detects and an average-rate comparison misses.

The implementation checks greedy reservations without explicitly enumerating all interval inequalities in routine operation; this is correct only for the stated interval eligibility sets. Its small-instance regression suite compares against a separate recursive assignment oracle on 171 exhaustive instances; 1,200 randomized calendars additionally compare the interval test. These checks support correctness of the implementation, not a replacement for the proof or validation of physical assumptions.

\subsection{Tight Aligned-Batch Bound and Retry Qualification}
\begin{corollary}[Aligned batches]\label{cor:batch}
Suppose every stream emits one SDU at $\ell H$, $F_i=H$, the common effective relative deadline $\bar D\le H$, and all $Q$ allocated slots are contiguous from the anchor. With no cross-period work, worst-case reservation is feasible exactly when
\begin{equation}\label{eq:batch}
 \sum_i m_iK_i\le\min\{Q,\lfloor\bar D/C\rfloor\}.
\end{equation}
\end{corollary}
\begin{proof}
All SDUs have the same available prefix of the period, containing exactly the number of slots on the right-hand side. Necessity is counting. Sufficiency follows by assigning the required quanta consecutively from the anchor in any job order. The construction repeats without carry because $\bar D\le H$.
\end{proof}
This is a tight bound only for the stated alignment and deadline assumptions. It is not the admission test used to certify the random-phase simulation, and it is not an infinite-horizon result for unequal periods.

A worst-case reservation guarantees a scheduled number of exchange \emph{opportunities}, provided measured service durations fit their assigned slots; it does not guarantee decoding success under arbitrary RF losses. Conditional on each fragment succeeding within its cap, all reserved SDUs can complete by deadline. Under independent attempt failures, an individual SDU then has the probability in~\eqref{eq:iid}. A joint lower bound for a reserved finite job set is
\begin{equation}\label{eq:union}
 \Prb\{\text{all jobs succeed in time}\}\ge
 \max\left\{0,1-\sum_jm_jp_j^{K_j}\right\}.
\end{equation}
The union-bound step does not require independence between jobs, but the bound $p_j^{K_j}$ on a fragment's exhaustion does require a valid conditional or independent-loss assumption. Burst correlation invalidates substituting its stationary mean into that expression. The online heuristic does not reserve $mK$ quanta and is not assigned~\eqref{eq:union} as a guarantee.

\subsection{Aggregate Capacity and Per-Stream Admission}
The exact reservation theorem characterizes a particular finite set of deadline windows. An independent work-counting argument gives an immediate test for whether a proposed population of streams can ever meet an aggregate on-time target. It applies even to a scheduler with oracle knowledge of future failures and perfect switching between flows.

\begin{proposition}[Necessary aggregate completion bound]\label{prop:capacity}
Consider $N$ streams with common SDU period $F$, each successful SDU requiring at least $m$ occupied exchange opportunities, and a stationary calendar with $Q$ opportunities every $H$ seconds. If all scheduled attempts share this service pool, the steady-state aggregate fraction $x$ of successfully completed SDUs satisfies
\begin{equation}\label{eq:capacity}
 x\le\min\left\{1,\frac{QF}{NmH}\right\}.
\end{equation}
For heterogeneous streams, if $x_i$ is the completion fraction of stream $i$, then $\sum_i(m_i/F_i)x_i\le Q/H$.
\end{proposition}
\begin{proof}
Over $T$ seconds, successful SDUs require at least $\sum_i m_i n_i(T)$ occupied opportunities, regardless of retry policy. The periodic calendar supplies at most $(Q/H)T+O(1)$ opportunities. Dividing by $T$ and taking the long-run limit yields the heterogeneous inequality. For identical $m,F$, $\sum_i x_i=Nx$ gives~\eqref{eq:capacity}. All missed deadlines and failed exchanges can only reduce the attainable fraction, so the bound is necessary, never sufficient.
\end{proof}

At the baseline setting, $Q=8$ and $H=5$ ms provide 1,600 opportunities/s. The eight-byte format uses three fragments per 240-byte SDU; eight 100-frame/s streams need at least 2,400 exchanges/s \emph{before any retransmission}. Even perfect scheduling and a loss-free channel cannot exceed an aggregate completion fraction of $2/3$ in steady state. The 12-stream limit is $4/9$. The recorded pruning means (59.399\% and 39.595\%) are below but relatively close to these capacity ceilings; their high utilization is not evidence that any individual stream receives acceptable service.

Table~\ref{tab:capacity} also quantifies the consequence of adopting the distinct 24-byte group envelope in Section~III. Four fragments per SDU lower the loss-free upper bound to 50\% for eight streams and 33.33\% for twelve streams on the \emph{same} calendar. These are analytical envelope sensitivities, not rerun success measurements. A 90\% minimum completion target for every stream is therefore infeasible at either load even before channel losses. Admission, payload-budget revision, or more opportunities is a prerequisite to a meaningful per-stream QoS guarantee.

\begin{table}[t]
\caption{Loss-free aggregate capacity ceilings for baseline workload; measured column is the original eight-byte simulation.}
\label{tab:capacity}\centering\small
\begin{tabular}{cccc}\toprule
Streams $N$ & 8-B ceiling & 24-B ceiling & 8-B pruning\\\midrule
4 & 100.00\% & 100.00\% & 90.594\%\\
8 & 66.67\% & 50.00\% & 59.399\%\\
12 & 44.44\% & 33.33\% & 39.595\%\\\bottomrule
\end{tabular}
\end{table}

A per-flow target $x_i\ge\eta_i$ further requires $\sum_i m_i\eta_i/F_i\le Q/H$; window constraints from Theorem~\ref{thm:hall} are stronger. This necessary test is deliberately independent of frame loss, clock uncertainty, queueing discipline, or an unverified quality-of-experience threshold. The combined admission strategy is to reject aggregate impossible demands first, then apply a validated finite-window reservation test on the actual slot-to-connection eligibility relation. It is inappropriate to admit every offered stream and subsequently describe a high mean completion fraction as a service guarantee.

\subsection{Flush Budget as a Calendar Constraint}
A useful UWB analogue of flushing is a bound on the last permitted attempt completion, not an unqualified multiplication of two Bluetooth parameters. Let $A_j=a_j-r_j$ be anchor delay and $V_j(n)$ the elapsed time from eligibility until the completion of the $n$th reserved opportunity for job $j$, including calendar gaps and competing reservations. A sufficient budget constraint is
\begin{equation}\label{eq:flush}
 A_j+V_j(m_jK_j)\le D_i-\epsilon_i.
\end{equation}
If a negotiated event lifetime allows $f$ calendar periods and a final within-period completion offset $O_j$, an additional explicit condition is $A_j+(f-1)H+O_j\le D_i-\epsilon_i$. The opportunity allocation must still provide $m_jK_j$ services before that time. A retry cap without an eligible-window count is insufficient; an event cap without anchor alignment and final-event offset is also insufficient.

Clock uncertainty can be incorporated by shrinking windows, provided the bound is genuine and conservatively covers all relevant timestamp/dispatch errors. It cannot be justified by a nominal oscillator specification alone. If a hardware exchange sometimes spills past $C$, the unit-slot assumptions fail, and neither the theorem nor an empirically observed zero late-attempt count repairs that mismatch.

\section{Online Scheduling and Finite Design Search}
\subsection{Deadline-Aware Optimistic Pruning}
Algorithm~\ref{alg:prune} is the executed online heuristic. At slot $k$, let $z_j$ be the fragments still required and $N_k(d_j)$ the number of remaining opportunities finishing by $d_j$. When $N_k(d_j)<z_j$, even an otherwise empty channel with no future losses cannot complete that SDU. Discarding it avoids spending an exchange that cannot improve its timely complete-SDU outcome. This is an optimistic impossibility test: it ignores competition, so passing it is not evidence of feasibility. The online policy operates at the abstract opportunity selector; the configured vendor slot graph is treated separately in Section~VIII.

\begin{algorithm}[t]
\caption{Online SDU pruning and EDF service}\label{alg:prune}
\begin{algorithmic}[1]
\For{each usable opportunity $[s_k,s_k+C]$}
 \State Admit SDUs with $a_j\le s_k$
 \For{each active SDU $j$}
  \If{$s_k+C>d_j$ \textbf{or} $N_k(d_j)<z_j$}
   \State Drop $j$ and report a frame-loss event
  \EndIf
 \EndFor
 \If{an active SDU remains}
  \State Select earliest $(d_j,r_j,\mathrm{streamID})$
  \State Attempt its current fragment; charge one full $C$
  \If{exchange succeeds}
   \State Advance fragment; complete SDU if $z_j=0$
  \ElsIf{fragment has used $K_i$ total attempts}
   \State Abort the SDU
  \EndIf
 \EndIf
\EndFor
\end{algorithmic}
\end{algorithm}

A hard admission/reservation mode is necessary when each stream needs a defined attempt-service guarantee, and additional loss assumptions are required for probabilistic completion guarantees. The evaluated pruning mode remains a stress-study policy: it has neither per-flow reservation nor a minimum service-rate constraint, so starvation is an admissible outcome. In the present common-relative-deadline workload, true arrival-order FIFO with the same expiry logic is exactly EDF without pruning. We enforce this equivalence in a regression test. Consequently, any measured difference between those two curves would indicate an implementation or tie-order mismatch, not an EDF benefit.

The reference implementation uses an explicit opportunity calendar and a list of active SDUs; it makes a scheduling decision before each simulated opportunity, whereas the vendor stack can already contain queued or in-flight PDUs. With $M$ slots and at most $J$ active jobs, repeated binary-search counting and scans give an $O(MJ\log M)$ upper bound for pruning; the no-pruning version and simple reservation checker are $O(MJ)$. Calendar storage is $O(M)$, with job storage proportional to the finite input set. These are inspectable Python implementation bounds, not measured MCU execution times. More efficient deadline data structures and an incremental supply count are possible, but their runtime behavior is not evaluated here.

\subsection{Finite Transport-Budget Search}
The executed search enumerates
\begin{align*}
 F&\in\{7.5,10\}\ms,\quad H\in\{2.5,5,7.5\}\ms,\\
 P&\in\{48,80,96,120\}\ \mathrm{B},\quad K\in\{1,2,3,4,5\},\\
 b&\in\{64,96,192\}\ \mathrm{kbit/s}.
\end{align*}
There are $2\cdot3\cdot4\cdot5\cdot3=360$ candidates. Two identical synchronous streams release jobs over 60 ms; each candidate uses its own effective deadline from~\eqref{eq:budget}. The objective is
\begin{equation}\label{eq:opt}
 \min_{\theta\in\Theta}\ U(\theta)=\frac{2m(\theta)K(\theta)C}{F(\theta)},
\end{equation}
subject to finite-calendar reservation feasibility, slot-fit constraints, $R_{\rm iid}\ge0.99$ at $p=0.1$, and $b\ge96$ kbit/s. Ties prefer higher bitrate, then smaller $F,H,P$, in that order. All candidates, not just the winner, are retained. This is exhaustive only over the stated finite grid and finite set of release jobs. It does not prove an optimum across arbitrary phases, an infinite stream, continuous design variables, or codec profiles absent from the grid.

The finite search uses a transport-cost objective rather than a subjective audio score. This choice makes its decision variables and optimum auditable: the SDU size is obtained from $(b,F)$, the number of fragments from $(P,\,\mathrm{header\ size})$, and the reserved exchange load from $(m,K,C)$. The reliability constraint uses an explicit independent-failure probability rather than an unexplained quality penalty. A codec-constrained optimization would constitute a different objective, because decoding distortion and concealment are not functions of the exchange budget alone. Hence the selected point is interpreted strictly as the optimum of the stated transport design space; no perceptual audio score is imputed to it.

\section{Simulation Methodology and Results}

\subsection{Code-to-Claim Traceability}
Table~\ref{tab:traceability} maps the results to their executable origin. This distinction is important because the core model, byte-level adaptation, and compiled probe are separate programs rather than a single measured audio stack. The revised companion contains 28 unit tests: 19 retained tests covering the existing scheduler, single-context envelope and capture parser, plus nine group-envelope tests. The reservation checks include 171 small exact-oracle instances, and the archive records a further 1,200-case randomized comparison of the interval test and assignment methods. The 1,200-case archive is source provenance, not a new hardware experiment. Passing these tests does not establish that the calendar, exchange service, or presentation clock matches a radio device.
\begin{table*}[t]
\caption{Source-level evidence and limitations of the principal manuscript claims.}\label{tab:traceability}\centering\small
\begin{tabular}{p{2.7cm}p{5.5cm}p{8.0cm}}\toprule
Claim & Executable source and recorded output & Scope restriction\\\midrule
Exact finite admission & \texttt{model.py:reserve\_jobs}; \texttt{tests/test\_model.py}; \texttt{theory-validation.json} & Equal-duration exchange quanta and finite known windows; no PHY validation or arbitrary-phase infinite-horizon proof.\\
Online pruning benefit & \texttt{model.py:simulate}; \texttt{experiments.py}; \texttt{replicates.csv}; \texttt{paired\_comparisons.csv} & Shared modeled calendar and slot-indexed losses; online heuristic is not an implementation of the conservative reservation theorem.\\
Worst-stream risk & \texttt{validation.py}; \texttt{per\_stream\_fairness.csv}; \texttt{fairness\_replications.csv} & Model streams are traffic sources, not independently deployed receiver devices; minimum-stream rates expose unacceptable overload QoS.\\
SDU fragmentation & \texttt{transport.py}; \texttt{tests/test\_transport.py} & Original 8-byte receiver keyed only by sequence; conditional firmware compatible.\\
Group identity/expiry & \texttt{multistream\_transport.py}; \texttt{test\_multistream\_transport.py} & Separate 24-byte Python reference, nine new tests; not part of earlier plots, hardware, or any authenticated clock protocol.\\
Probe firmware & \texttt{hardware/iso\_probe.h}; build manifests/logs & Cross-compiled known-pattern probes, not flashed or tested on RF or audio hardware; SDK slot ownership differs from unrestricted simulator dispatch.\\\bottomrule
\end{tabular}
\end{table*}

\subsection{Executed Design, Baselines, and Statistical Unit}
The main grid contains 48 one-factor cases, six policy variants, and 12 seeds per case, totaling 3,456 executed runs. Each produces six seconds of source traffic followed by a fixed drain horizon. The cases cover logical-stream count $1,2,4,6,8,10,12$ under IID and burst losses; exchange error probabilities $0,0.05,0.1,0.2,0.3$ in both channels; generic frame lengths $7.5,10,15,20$ ms; superframes $2.5,5,7.5$ ms; PDU budgets $48,80,96,120$ bytes; caps $1$--$5$ total attempts; guards $0,20,40,80\us$; and exchange durations $120,160,200,300\us$. The design is a one-factor sensitivity study rather than a full factorial design.

The six policies are (1) one-shot FIFO, (2) bounded-retry FIFO without expiry, (3) bounded-retry FIFO with expiry, (4) work-conserving per-stream periodic round robin (RR), (5) EDF with expiry but without pruning, and (6) Algorithm~\ref{alg:prune}. All methods except the one-shot diagnostic use the configured total-attempt cap; thus, the meaningful overload attribution is between the expiry-aware variants, not the artificially under-provisioned one-shot baseline. Within each common-condition, same-seed comparison, policies share arrival phases, opportunity indices, and a channel trace drawn independently of dispatch choices. The trace continues to evolve during idle opportunities; this controls common randomness but does not validate a physical fading process. One-shot FIFO has a lower attempt cap by definition; comparisons with it are not equal-reliability designs. Bounded FIFO without expiry is intentionally weak under overload and is not the sole baseline. Periodic RR is only a same-PHY abstraction inspired by periodic service, not an implementation of CIS/BIS, FT, BN, or controller packing.

One independent pseudorandom seed constitutes an experimental replication, not each SDU within that run. We report $\bar x\pm t_{0.975,11}s_x/\sqrt{12}$ for seed-level estimates across the 12 replications in a selected scenario. Policy differences use paired seeds; marginal interval overlap is not used as a significance test. These two-sided intervals are exploratory, pointwise 95\% intervals without a family-wise multiple-testing adjustment; no post hoc selection of significant comparisons is used to establish universal ranking. Under heavy congestion, a rate over six seconds describes the stated model and run length rather than a stationary service guarantee. The full main data contain 9,504,000 modeled SDU job evaluations across repeated policies/scenarios, not that many independent physical measurements. Another 144 runs quantify per-stream fairness and 288 study packet-size/PER sensitivity. Raw SDU and attempt traces are retained for the declared burst-loss, four-stream, seed-zero case; every run retains its seed-level aggregate row.

All non-timely SDUs count as misses, including those left in a no-expiry queue at the terminal drain horizon. The original result ledger marks 429,353 end-of-horizon censored SDUs, mostly in backlog-prone policies; the exact count follows the retained six-second sources and fixed drain horizon. This is a bookkeeping classification, not right-censoring in a survival-analysis model. We report this rather than silently discard the backlog. Loss-only or latency-only comparisons would be misleading because successful-SDU quantiles condition on the subset that completed before deadline.

\subsection{Timely Completion: Gains and a Negative Result}
Figure~\ref{fig:load} shows the strong-policy comparison. At four burst-loss streams, expiry-aware FIFO, no-pruning EDF and pruning EDF all attain 90.594\% mean timely completion with a 0.834 percentage-point interval half-width. Periodic RR attains 91.698\% with half-width 0.915. The paired pruning-minus-RR difference is $-1.104\pm0.268$ percentage points across the twelve seed pairs, with the stated interval convention. The result rules out a claim of uniform dominance: interleaved per-stream service can be advantageous when losses exhibit the assumed opportunity-index correlation, even though the periodic policy is much weaker under heavier load.

At eight streams, pruning reaches 59.399\% versus 45.441\% for expiry-aware FIFO/no-pruning EDF, a paired gain of $13.958\pm2.080$ percentage points. At twelve streams, the corresponding values are 39.595\% and 27.264\%, a gain of $12.331\pm3.037$ percentage points. These absolute differences describe additional \emph{complete} SDUs served on time, not a change in bit-error rate, codec artifacts, or receiver audio quality. Table~\ref{tab:results} makes the reporting unit explicit. Periodic RR can spread partial SDU work across so many streams that few whole SDUs finish before expiry. Its collapse in this overloaded abstraction does not imply that a correctly configured Bluetooth controller or admitted LE Audio session behaves similarly.

\begin{figure*}[t]\centering
\includegraphics[width=.96\textwidth]{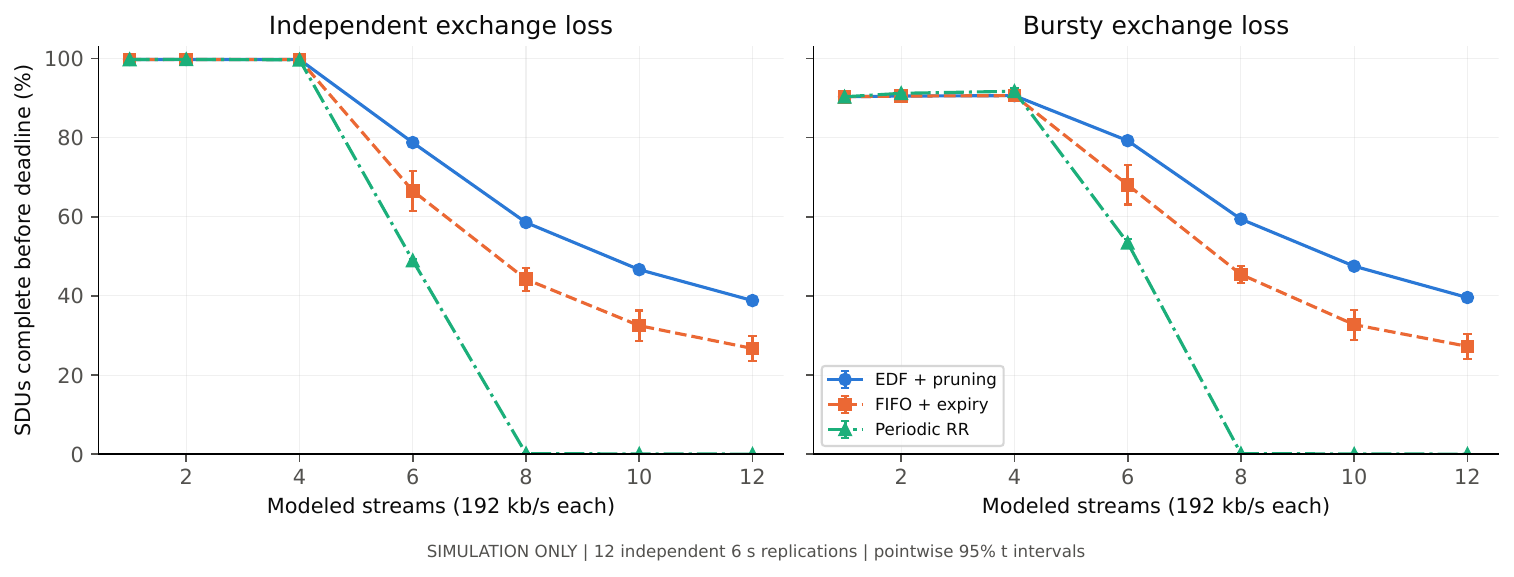}
\caption{Modeled results. Timely complete-SDU fraction versus logical-stream count, using the same UWB calendar and paired channel traces. Each point summarizes twelve independent six-second runs; bars are pointwise 95\% seed-level $t$ intervals. Periodic RR is not a standards-conformant LE Audio benchmark.}\label{fig:load}
\end{figure*}

\begin{table}[t]\caption{Burst-loss timely completion (\%), mean $\pm$ 95\% interval half-width in percentage points.}\label{tab:results}\centering\small
\begin{tabular}{rrrr}\toprule
Streams & EDF + pruning & FIFO + expiry & Periodic RR\\\midrule
4 & $90.594\pm0.834$ & $90.594\pm0.834$ & $91.698\pm0.915$\\
8 & $59.399\pm0.651$ & $45.441\pm2.083$ & $0.189\pm0.045$\\
12 & $39.595\pm0.433$ & $27.264\pm3.150$ & $0.050\pm0.017$\\\bottomrule
\end{tabular}
\end{table}

These overload cases are adverse stress tests, not admissible QoS operating points. Even with zero errors, the nominal fragment demand at twelve streams is $12\cdot3\cdot200/10000=0.72$ of wall-clock time, whereas the calendar reserves only 0.32. A fair audio system would reject or reconfigure excess demand instead of advertise the aggregate result as guaranteed service. Worst-case attempt reservation is still more demanding, with $mK$ rather than $m$ work per SDU.

\subsection{Ablation, Starvation, and Resource Use}
Figure~\ref{fig:ablation} compares pruning with the two equivalent expiry-only policies. At common relative deadlines, FIFO and EDF order cannot explain a gain; pruning is the changed decision. The equality is checked in tests and in every final main-grid completion statistic. This avoids a common pitfall in which same-anchor batches accidentally reorder FIFO by stream ID and create a spurious EDF effect.

\begin{figure*}[t]\centering
\includegraphics[width=.96\textwidth]{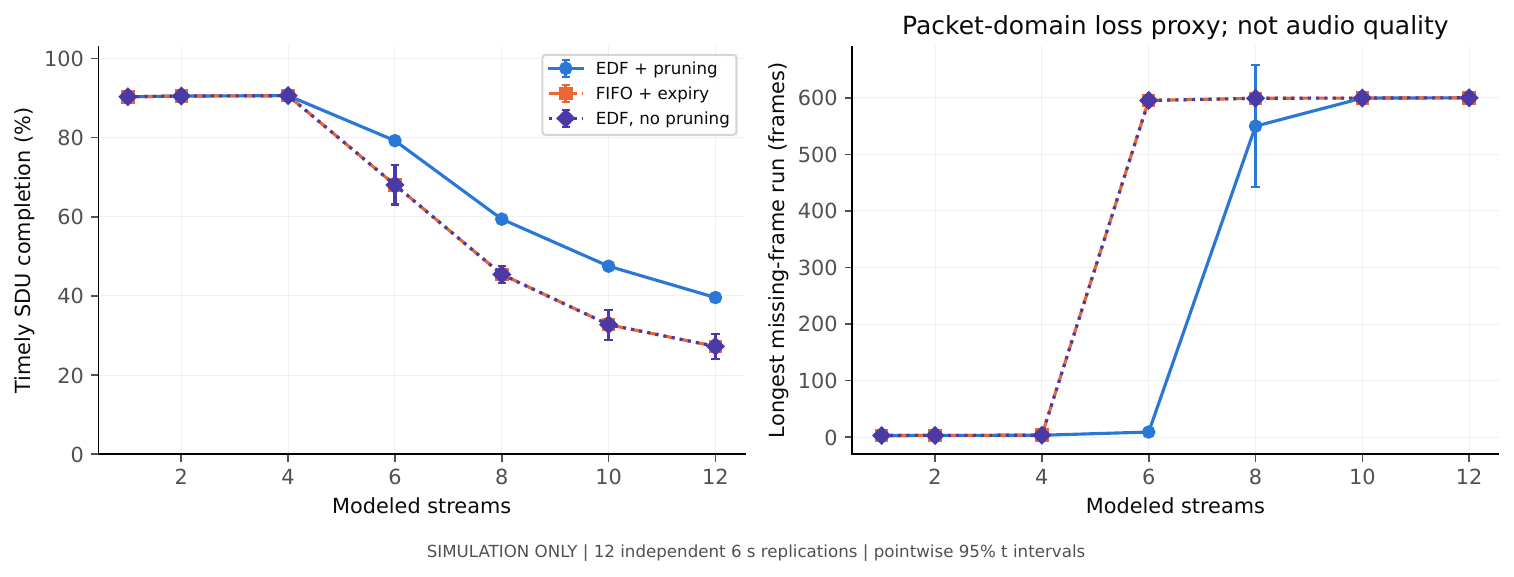}
\caption{Modeled results. Optimistic pruning ablation under burst losses. FIFO+expiry and EDF without pruning coincide in this homogeneous model. The missing-frame-run statistic is computed within each stream before taking a run maximum; it is not a PLC or listening-quality metric.}\label{fig:ablation}
\end{figure*}

Aggregate completion hides a serious limitation. In the supplemental eight-stream burst case, mean worst-stream completion for pruning is only 2.556\%, while its mean Jain index is 0.737, computed per run as $(\sum_i x_i)^2/(N\sum_i x_i^2)$ for per-stream completion fractions $x_i$. At twelve streams, the mean worst-stream completion is approximately 0.014\%, with Jain index 0.510. These are means of replication-level minima and indices, not a minimum of pooled means. A stream can therefore receive almost no timely SDUs despite an apparent aggregate efficiency gain. Figure~\ref{fig:robust} displays this failure, and Proposition~\ref{prop:capacity} explains why no scheduling policy can meet an arbitrarily high per-stream target in the specified overloaded configuration. Per-stream admission, explicit reservation, or controlled quality/rate adaptation is required; an average throughput objective alone is inadequate.

The occupied-slot fraction in Fig.~\ref{fig:latency} accounts for all attempted full opportunities, not only successful data bytes. It is useful for identifying wasted service, but not for ranking energy efficiency. The RX duty cycle, wakeup energy, CCA, and power-amplifier behavior are absent. No claim about battery lifetime follows from this figure.

\subsection{Parameter Sensitivity and Tail Interpretation}
Figure~\ref{fig:sensitivity} varies one factor at a time. Increasing frame duration at fixed presentation budget does two things simultaneously: it enlarges the SDU and reduces the remaining radio deadline. The 20 ms case leaves only 3 ms before the uncertainty guard and is not a pure test of MAC efficiency. The 15/20 ms values are generic batching assumptions, whereas LC3 operating points must be selected from actual legal codec configurations~\cite{bap,lc3white}.

PDU size changes the number of all-or-nothing fragments. At 192 kbit/s and 10 ms, a 240-byte SDU needs six fragments at $P=48$ but three at $P=96$. The primary sweep holds exchange PER fixed, which can favor large PDUs without modeling their RF error cost. The supplemental sensitivity therefore also uses the explicitly assumed surrogate
\begin{equation}\label{eq:sizeper}
 p(P)=1-(1-0.1)^{(P+20)/(96+20)}.
\end{equation}
The 20-byte equivalent fixed length is an assumption, not a measured header/BER calibration. For pruning, $P=48$ yields 57.372\% timely completion under fixed PER and 61.056\% under this lower-PER small-packet surrogate. At $P=120$, the values are 99.712\% and 99.531\%. The exercise exposes model dependence, not a calibrated optimum payload.

\begin{figure*}[t]\centering
\includegraphics[width=.94\textwidth]{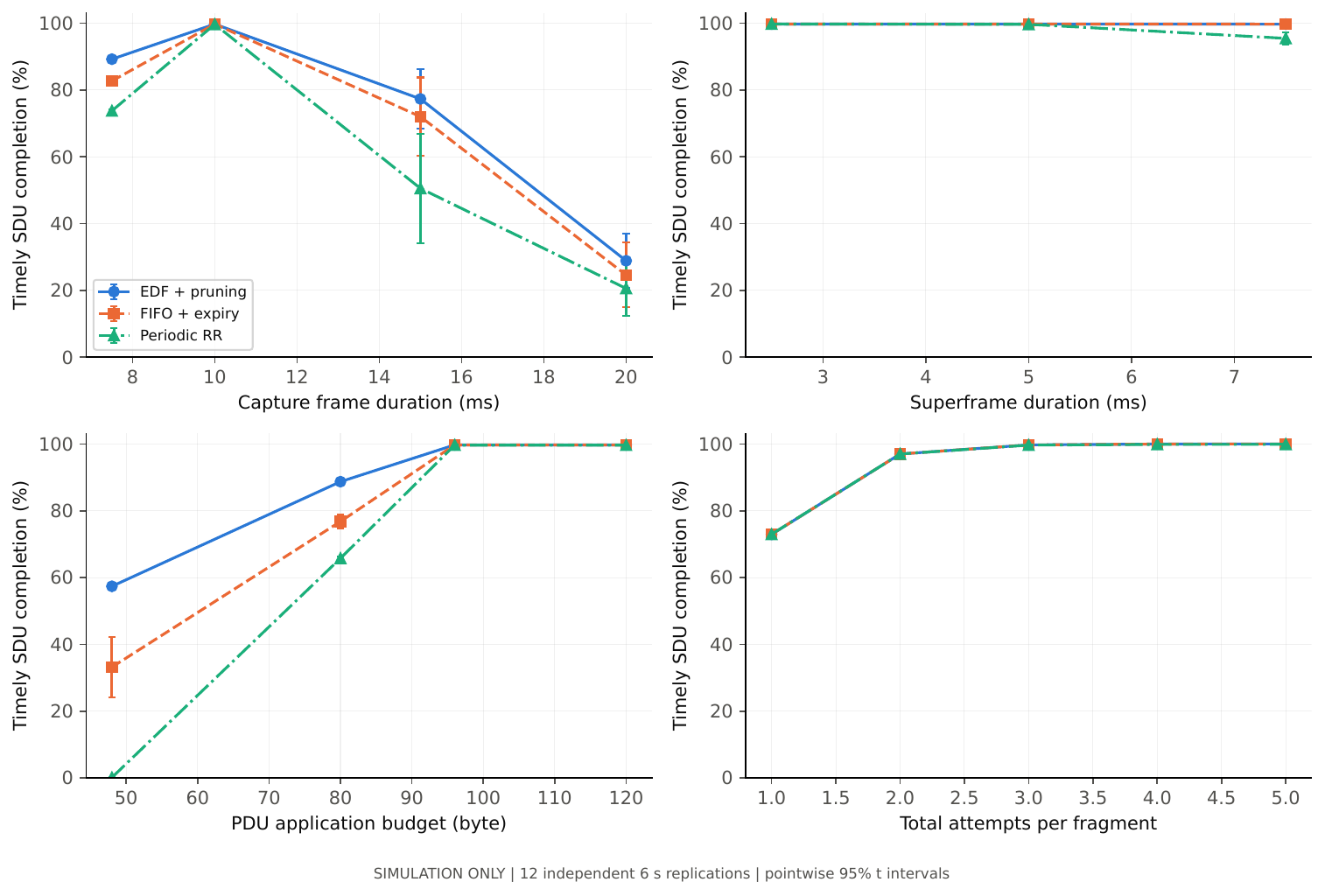}
\caption{Modeled results. One-factor frame, superframe, PDU-budget and total-attempt sweeps. Frame length changes capture and radio budgets together. Main size sweep holds exchange PER fixed; Fig.~\ref{fig:robust} supplies a separate assumed length-dependent-PER sensitivity.}\label{fig:sensitivity}
\end{figure*}

The superframe and slot sweeps also quantize capacity: $Q=\lfloor\rho H/C\rfloor$, so equal nominal share does not imply exactly equal usable share. For example, at $H=2.5$ ms and $C=0.2$ ms, $Q=4$, while at $H=7.5$ ms, $Q=12$. A different slot size changes both fragment-service rate and unused fractional opportunity budget. These are legitimate discrete design effects, but must not be attributed solely to latency.

Figure~\ref{fig:latency} reports the mean of each run's timely-SDU 99th percentile. For pruning at four burst-loss streams it is 5.888 ms; at eight and twelve streams it is 12.812 and 12.895 ms. These are not acoustic latency or an unconditional 99th-percentile service guarantee. More seriously, overloaded RR runs at eight streams have only 4--14 timely SDUs each, and at twelve streams only 0--6. Their displayed quantiles are descriptive order statistics, with an empty run omitted, not credible estimates of a population tail. The artifact preserves this raw figure for audit; we draw no tail superiority conclusion from these sparse-RR points. The pruning runs at those loads contain approximately 2,796--2,939 timely SDUs each, but even those quantiles remain conditioned on severe loss/starvation.

\begin{figure*}[t]\centering
\includegraphics[width=.96\textwidth]{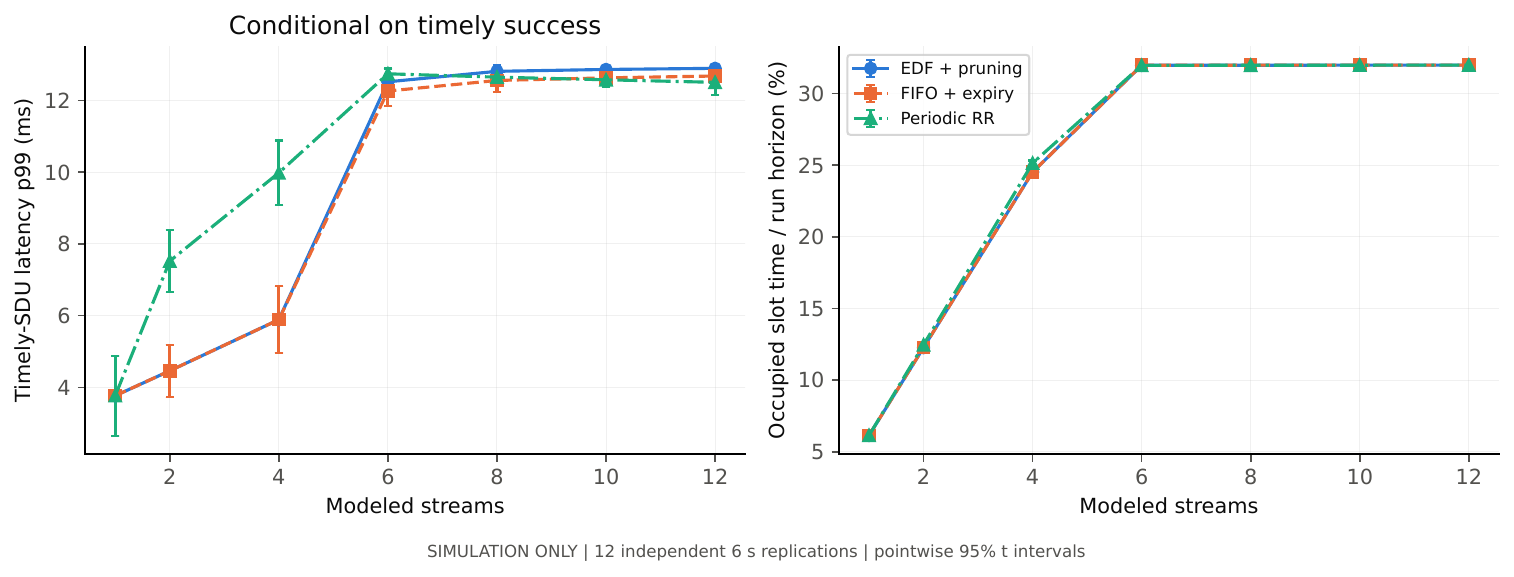}
\caption{Modeled results. Success-conditioned latency and occupied-slot fraction. RR tail points at eight/twelve streams have very few timely samples (4--14 / 0--6 per run); they are descriptive only, not robust p99 estimates. Missing SDUs are excluded from all shown latency quantiles, so the completion plot must be read alongside this figure. No measured energy or acoustic delay is shown.}\label{fig:latency}
\end{figure*}

\begin{figure*}[t]\centering
\includegraphics[width=.96\textwidth]{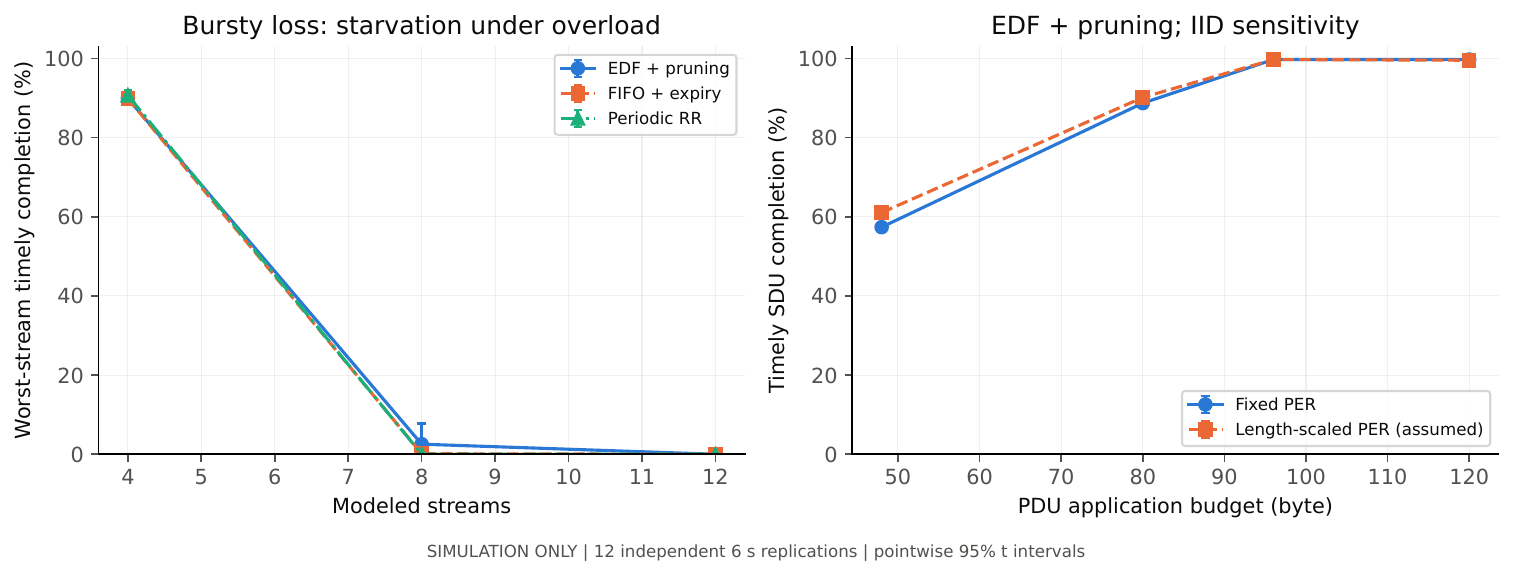}
\caption{Modeled results. Left: worst-stream timely completion, exposing overload starvation. Right: fixed-PER versus explicitly assumed length-scaled-PER sensitivity for pruning. Neither error model is a calibrated SPARK physical channel, and neither panel measures decoded audio quality.}\label{fig:robust}
\end{figure*}

\subsection{Finite Design Results and a Quantitative LE Timing Example}
The finite enumeration selects $F=7.5$ ms, $H=2.5$ ms, $P=120$ bytes, $K=2$, and $b=96$ kbit/s. An SDU is 90 bytes and fits one 112-byte body. Two streams therefore reserve four 200 $\mu$s opportunities per frame, giving $U=800/7500=10.667\%$; the uncongested IID success expression is $1-0.1^2=0.99$. Figure~\ref{fig:design} shows a declared slice of the full candidate set. The result is a finite-grid transport-budget optimum, not a demonstrated optimal LC3/PLC setting. Its 60 ms release horizon and synchronous phases cannot certify an infinite asynchronous stream.

\begin{figure*}[t]\centering
\includegraphics[width=.96\textwidth]{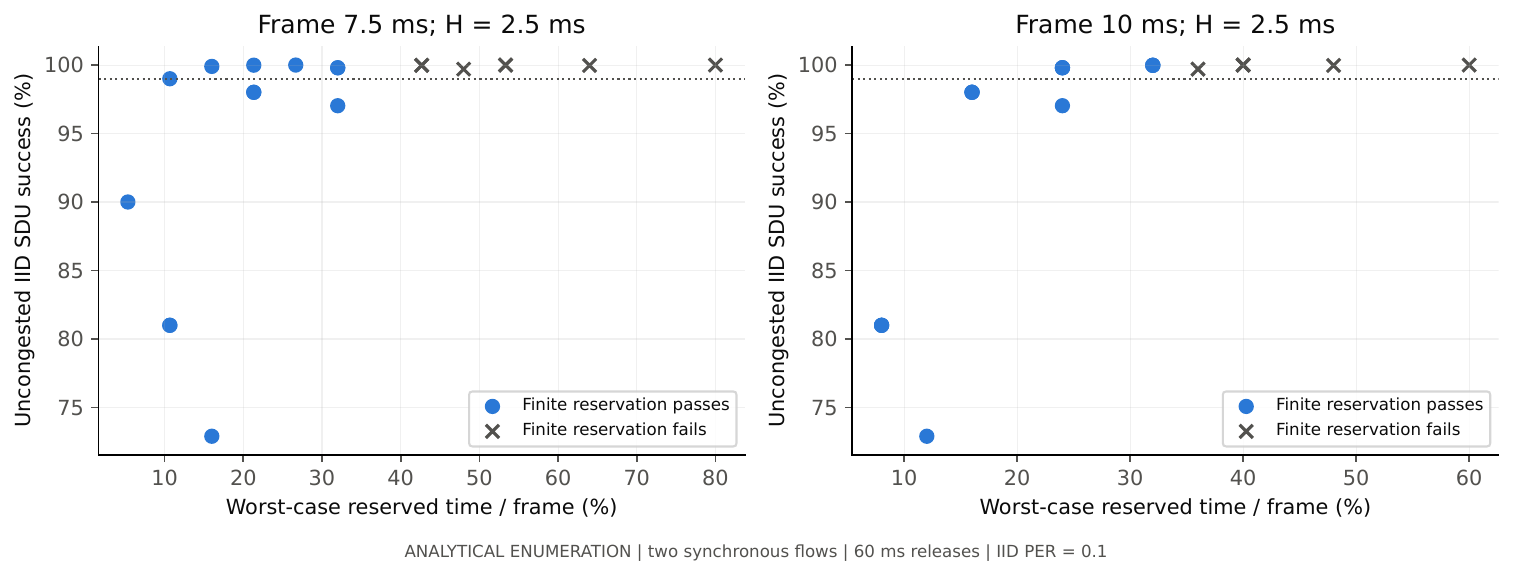}
\caption{Finite-grid enumeration. A 96 kbit/s, $H=2.5$ ms slice of the 360 candidates. Finite reservation passes/fails refer to the stated two-stream 60 ms release set, not to arbitrary-phase infinite-horizon schedulability. Reliability on the vertical axis is the uncongested IID expression, not measured packet delivery.}\label{fig:design}
\end{figure*}

A standards-based timing example clarifies the requested LE comparison without inventing BLE experiments. Take unframed CIS with $I=S=10$ ms, $U=1$, and an illustrative $G=2$ ms. Equation~\eqref{eq:bleunframed} gives 2 ms at FT=1 and 12 ms at FT=2. Simply using FT$\times I$ would instead give 10 and 20 ms, omitting the synchronization term and using the wrong reference. For segmentable framed CIS with the same $I,S,G$ and $F_{\rm framing}=S$,~\eqref{eq:bleframed} gives 22 ms at FT=1. These are algebraic illustrations, not asserted realizable controller configurations or measured latencies. They also exclude implementation costs as the standard specifies. Our 13 ms radio budget cannot be compared directly with a BAP Presentation Delay value without reconciling all reference points.

Table~\ref{tab:lemap} places the scheduling quantities alongside their LE Audio counterparts. The table is a semantic comparison, not a claim of protocol equivalence: CIS/BIS lifetimes, subevents, and presentation references have standard-specific definitions, whereas $H$, $Q$, $K$, and $d_j$ describe the UWB transport model. Consequently, the calculations above test the internal consistency of timing references without assigning a relative throughput, delay, or energy advantage to either technology.

\begin{table*}[t]\caption{LE Audio reference semantics versus the current UWB artifact. Parameters refer to distinct protocol layers.}\label{tab:lemap}\centering\small
\begin{tabular}{p{3.0cm}p{5.1cm}p{7.0cm}}\toprule
LE Audio feature & Verified interpretation & UWB model correspondence\\\midrule
CIS/CIG; BIS/BIG & Connected versus broadcast grouping with distinct feedback and timing~\cite{btll} & One-transmitter logical-stream service calendar.\\
ISO\_Interval / subevents & Event period and within-event opportunities, constrained parameters~\cite{btll} & Periodic $H,C,Q$ opportunity calendar, non-normative notation.\\
BN / NSE / FT & New payload/burst and opportunity/lifetime semantics for CIS~\cite{btll,isoal} & Per-fragment cap $K$ and absolute deadline $d_j$; distinct semantics.\\
BIS retransmissions & Broadcast repetition/pretransmission uses IRC/PTO and timing~\cite{btll,isoal} & Unicast ACK-based loss process; distinct from BIS repetitions.\\
ISOAL & Segmentation, SDU timing reconstruction and loss indication~\cite{isoal} & Independent packet envelope and receiver reassembly semantics.\\
Presentation Delay & Sink reception: reference-to-presentation; source-server direction differs~\cite{bap} & Application capture-to-presentation budget with explicit reference points.\\
LC3 / PLC & Legal codec modes and decoder behavior~\cite{bap,lc3white} & Fixed-size SDU traffic; transport-level deadline completion objective.\\\bottomrule
\end{tabular}
\end{table*}

\subsection{Artifact Organization and Reproducibility}
The source package distinguishes the calendar and channel model (\texttt{code/src/model.py}), experiment driver (\texttt{code/src/experiments.py}), and sensitivity and correctness checks (\texttt{code/src/validation.py}). Principal seed-level records include \texttt{replicates.csv}, \texttt{fairness\_replications.csv}, \texttt{size\_per\_sensitivity.csv}, \texttt{paired\_comparisons.csv}, and \texttt{optimization.json}. The reported figures are obtained from these records; model-based numerical curves are not device measurements. The eight-byte single-context transport and the distinct 24-byte group/stream/epoch implementation reside in \texttt{code/src/transport.py} and \texttt{code/src/multistream\_transport.py}, respectively. The 24-byte envelope supports disambiguated reassembly and tested expiry logic; all reported scheduler traces use the eight-byte format, whereas the 24-byte capacity ceilings are calculations from payload accounting.

Unit tests run with \texttt{PYTHONPATH=code/src python -m unittest discover -s code/tests -v}; the suite contains 28 tests across the calendar, transport, and trace parser. The supplied record also includes 171 small exact-oracle reservation checks and 1,200 randomized finite-calendar checks. These are algorithmic checks under the specified model. The accompanying code and data permit reproduction of the reservation checks, stochastic scheduling comparisons, parameter enumeration, and packet-envelope tests without attributing physical-layer performance to program execution.

\section{Controller Interface and Timing Semantics}
\subsection{Documented Service and Source-Level Integration}
The inspected SR1100 SDK 2.4.0 provides periodic slot configuration, ACK and retransmission controls, connection-specific deadlines, native fragmentation options, and transfer callbacks~\cite{spark}. The vendor's unidirectional audio example uses a 49-slot schedule with 200-$\mu$s slots, including 48 coordinator slots and one reverse slot, and demonstrates an existing PCM transport configuration. These mechanisms are controller features rather than new contributions of the proposed scheduling analysis. The inspected radio and PHY libraries are prebuilt, while the application and SDK header interfaces expose their configuration contracts.

The \texttt{swc\_connection\_set\_retransmission} interface specifies total transmission tries and a deadline expressed in backend HAL ticks. In the conditional Quasar backend examined here, the clock is a 1,000-Hz free-running counter, so 13 ticks correspond to 13~ms. This deadline is a PDU-related controller setting and is not interchangeable with an application SDU presentation deadline. Likewise, a configured slot list constrains which connection can issue a PDU: unrestricted opportunity reassignment in the mathematical model is a relaxation relative to connection-owned timeslots. The connection-to-slot eligibility relation is thus part of the service specification, not an incidental implementation detail.

The companion firmware source modifies the documented two-node Hello-World example with an eight-byte SDU envelope and a 240-byte test pattern released every 10~ms. It sends one fragment at a time, advances on an ACK-success callback, records drops and overlap skips, checks reassembled content, and reports complete-SDU events. Eight ELF/BIN artifacts were cross-compiled for the documented Quasar STM32U5A5/SR1100-family target with the recorded toolchain. The retained example schedule uses 500-$\mu$s transmit/reverse slots, giving one forward opportunity per millisecond rather than the grouped 200-$\mu$s opportunities of the primary numerical calendar. This difference prevents the compiled example's scheduling geometry from being mistaken for that of the reported numerical study.

\begin{table}[t]
\caption{Reproducible artifact classes and the associated claims.}\label{tab:hardware}\centering\small
\begin{tabular}{p{3.2cm}p{4.5cm}}\toprule
Artifact & Established scope\\\midrule
Reservation model & Mathematical feasibility and unit tests\\
Scheduling runs & Seed-level outcomes for specified opportunity and loss models\\
Eight-byte reassembly & Executed single-context byte semantics\\
24-byte envelope & Executed keyed multi-stream reassembly tests\\
Controller inspection & Documented slot and retransmission interface semantics\\
Firmware probe & Cross-compilation for the named target/toolchain\\\bottomrule
\end{tabular}
\end{table}

\subsection{Timestamp and Outcome Definitions}
The analysis distinguishes SDU release, fragment queue insertion, radio service, complete reassembly, and presentation. Let source release and receiver completion be $t^{(s)}$ and $t^{(r)}$. A one-way time difference is meaningful only when their time references are related by an identified offset, drift model, and uncertainty bound. In the analytical model, this uncertainty reduces the usable deadline by $\epsilon_i$; no synchronization quality is inferred merely from nominal tick resolution. The 16-bit millisecond field of the probe format wraps at 65.536~s and serves as a local reference rather than a globally synchronized timestamp.

An exchange attempt is distinct from a queue-enqueue callback, and an ACK-success indication is distinct from end-to-end decoded audio delivery. The parser associated with the probe records release, enqueue, ACK-success, ARQ-drop, expiry, complete-SDU, and buffer/error events. It checks sequence continuity, malformed fields, ring overflow, and timebase resets. These observability distinctions explain why the numerical objective is \emph{deadline-complete SDUs in a modeled exchange calendar}, whereas measured sound reproduction is a separate endpoint with codec and playout semantics. They also prevent accidental substitution of a transmitter-local controller clock for a receiver presentation reference.

\subsection{Protocol Overhead and Scheduler Realizability}
On a freely assignable calendar the EDF construction in Theorem~\ref{thm:hall} provides a complete finite reservation whenever the interval inequalities hold. When actual slots are tied to connections, candidate eligibility changes to a connection--slot graph and the interval theorem applies only after a compatible service pool has been identified. The online pruning rule is operationally meaningful at an application queue boundary where an SDU's remaining fragment work can be canceled; it is not an assertion that internal controller PDUs can be forcibly withdrawn at every opportunity.

The group-envelope design highlights a second interface constraint. An application using the 24-byte identity/epoch/presentation header has $72$ payload bytes per 96-byte PDU, compared with $88$ bytes for the original eight-byte header. The resulting increase from three to four fragments per 240-byte SDU changes the capacity bound before retries are considered. Thus, admission parameters must be defined for the selected encapsulation, release contract, and allowable slot graph. A capacity certificate for one wire format does not carry over to another without recomputing demand.

\section{Discussion}
\subsection{Aggregate Recovery Versus Receiver-Level Fairness}
The principal numerical effect of pruning is to reclaim capacity that expiry-aware FIFO or EDF without pruning can spend on SDUs already unable to finish. In the homogeneous workload, FIFO with expiry and EDF without pruning share the same deadline order, so EDF ordering in isolation produces no measurable improvement. The improvement observed at eight streams is instead attributable to excluding unfinishable partial SDUs. This attribution follows from the policy ablations and equality checks; it is more precise than describing the mechanism as an independently improved EDF policy.

An aggregate improvement does not certify a useful service level for every receiver. When offered fragments exceed the opportunity budget, completion must be sacrificed even on a perfect channel. At eight streams, the loss-free ceiling for the eight-byte format is $2/3$, and the lowest stream completion fraction observed in the burst-loss experiment is 2.56\%. At four streams, periodic round robin instead slightly exceeds pruning in overall timely completion. These findings make per-stream admission and scheduling fairness first-class quantities rather than secondary plots. The necessary heterogeneous capacity condition, $\sum_i(m_i/F_i)x_i\le Q/H$, is a useful screening test, although its satisfaction alone cannot overcome deadline-window or eligibility-graph restrictions.

\subsection{Interpretation of the Numerical Evidence}
The evaluation uses periodic fixed-size SDU traffic and two specified exchange-failure models. Independent-attempt loss permits the closed-form fragment expression, whereas the burst model evolves with calendar opportunities and produces correlated failures. Such models isolate scheduling effects under reproducible conditions, but loss parameters are model inputs rather than a device-specific propagation law. Success-conditioned latency excludes missed SDUs and is therefore not a stand-alone reliability measure. Run-level uncertainty comes from 12 pseudorandom replications per condition; individual comparisons have pointwise intervals, with no asserted family-wise coverage across all explored parameters.

The interval-demand theorem concerns finite known jobs with equal exchange quanta, contiguous eligibility windows, and freely assignable slots. If connection ownership, fragment dependencies, nonuniform service times, or dynamic clock uncertainty alter these assumptions, eligibility and demand must be represented in the corresponding constrained reservation model. Similarly, the finite 60-ms design enumeration establishes performance of its stated candidate grid, not an unrestricted global optimum. The controller example, keyed reassembly tests, and opportunity-level simulation address different layers of the proposed audio service; their reported outcomes retain those distinctions. This disciplined separation permits the capacity bound, numerical scheduling comparisons, and protocol overhead analysis to be read together without equating any of them to recorded acoustic performance.

\subsection{Design Implications}
Three implications follow directly. First, admission screening precedes comparison of online policies whenever offered work violates the calendar's structural capacity. Second, a scheduler that cannot complete a frame may free service for other SDUs, but a fairness-aware objective is required if every admitted stream has a minimum timely-completion target. Third, the number of fragments depends on stream-identification and timestamp metadata, so group protocol design and scheduling cannot be optimized independently. The central engineering decision is therefore a joint selection of the calendar, eligible connections, SDU envelope, attempt budget, and acceptable per-stream demand.

\section{Conclusion}
We formulated complete-SDU deadline delivery over a finite UWB opportunity calendar and separated conservative attempt reservation from online overload control. The finite interval-demand criterion and its EDF construction specify the exact reservation result under unit-slot, unrestricted-eligibility assumptions, while the aggregate bound exposes impossible offered loads without invoking a channel error model. Deadline-aware pruning recovers service that would otherwise be spent on frames incapable of timely completion; in the eight-stream burst-loss setting it improves aggregate completion by 13.958 percentage points over expiry-aware FIFO. The four-stream round-robin advantage and severely unbalanced eight-stream outcome show why aggregate salvage cannot replace per-stream service accounting. Finally, fragmentation-header overhead changes both the job workload and the provable calendar capacity. Joint consideration of admission windows, SDU identity, exchange budgets, and fairness is consequently essential to deadline-aware UWB audio transport.

\bibliographystyle{IEEEtran}
\bibliography{refs/references}

@article{coppens2022,
  author = {Dieter Coppens and Adnan Shahid and Sam Lemey and {Van Herbruggen}, Ben and Chris Marshall and {De Poorter}, Eli},
  title = {{An Overview of UWB Standards and Organizations (IEEE 802.15.4, FiRa, Apple): Interoperability Aspects and Future Research Directions}},
  year = {2022},
  doi = {10.1109/access.2022.3187410},
  journal = {IEEE Access},
  volume = {10},
  pages = {70219--70241}
}

@inproceedings{imfeld2023,
  author = {Jonah Imfeld and Silvano Cortesi and Philipp Mayer and Michele Magno},
  title = {{Evaluation of a Non-Coherent Ultra-Wideband Transceiver for Micropower Sensor Nodes}},
  year = {2023},
  doi = {10.1109/sensors56945.2023.10325219},
  booktitle = {Proc. IEEE SENSORS},
  pages = {1--4}
}

@article{yan2024,
  author = {Zhongjiang Yan and Hao Xu and Zhong Shen},
  title = {{Modeling and Analysis of the Performance for CIS-Based Bluetooth LE Audio}},
  year = {2024},
  doi = {10.1109/jiot.2024.3452187},
  journal = {IEEE Internet of Things Journal},
  volume = {11},
  pages = {40603--40618},
  number = {24}
}

@article{yan2026,
  author = {Zhongjiang Yan and Jiawei Bai and Xin Pan},
  title = {{Modeling and Analysis of Isochronous Communication Under Lifetime- and Flushing-Based Mechanisms}},
  year = {2026},
  doi = {10.1109/jiot.2026.3689201},
  journal = {IEEE Internet of Things Journal},
  volume = {13},
  pages = {31967--31981},
  number = {14}
}

@inproceedings{mantis2026,
  author = {Theo Gasteiger and Markus Schuß and Carlo Alberto Boano and Kay Römer},
  title = {{MANTIS: Towards Multiparty Communication for Bluetooth LE Audio}},
  year = {2026},
  doi = {10.1109/wfcs67029.2026.11511634},
  booktitle = {Proc. IEEE WFCS},
  pages = {1--8}
}

@inproceedings{wile2025,
  author = {Lingang Li and ZiLong Liu and Yihao Zhang and Siyuan Peng and Zhirong Liu and Yongrui Chen and Zhijun Li},
  title = {{WiLE-Audio: Wide-Coverage Low-Energy Audio via WiFi-BLE Cross-Technology Communication}},
  year = {2025},
  doi = {10.1145/3680207.3765258},
  booktitle = {Proc. ACM MobiCom},
  pages = {969--983}
}

@inproceedings{bruce2025,
  author = {Ian C. Bruce and Steve Armstrong and Daniel J. Bosnyak and Hany Tawfik},
  title = {{Opportunities and Challenges for Bluetooth LE Audio Assistive Listening Systems}},
  year = {2025},
  doi = {10.1109/icassp49660.2025.10889430},
  booktitle = {Proc. IEEE ICASSP},
  pages = {1--5}
}

@inproceedings{rtwt2024,
  author = {Andrey Belogaev and Xiaoman Shen and Chun Pan and Xingfeng Jiang and Chris Blondia and Jeroen Famaey},
  title = {{Dedicated Restricted Target Wake Time for Real-Time Applications in Wi-Fi 7}},
  year = {2024},
  doi = {10.1109/wcnc57260.2024.10571278},
  booktitle = {Proc. IEEE WCNC},
  pages = {1--6}
}

@inproceedings{fips2025,
  author = {Simon Egger and James Gross and Joachim Sachs and Gourav P. Sharma and Christian Becker and Frank Dürr},
  title = {{End-to-End Reliability in Wireless IEEE 802.1Qbv Time-Sensitive Networks}},
  year = {2025},
  doi = {10.1109/iwqos65803.2025.11143396},
  booktitle = {Proc. IEEE/ACM IWQoS},
  pages = {1--10}
}

@inproceedings{baghel2023,
  author = {Lalit Kumar Baghel and Radhika Raina and Suman Kumar},
  title = {{Evolution of Bluetooth Classic Audio towards Bluetooth LE Audio: Challenges and Road Ahead}},
  year = {2023},
  doi = {10.1109/iciis58898.2023.10253545},
  booktitle = {Proc. IEEE ICIIS},
  pages = {459--464}
}

@article{hou2012,
  author = {I-Hong Hou and P. R. Kumar},
  title = {{Queueing systems with hard delay constraints: a framework for real-time communication over unreliable wireless channels}},
  year = {2012},
  doi = {10.1007/s11134-012-9293-y},
  journal = {Queueing Systems},
  volume = {71},
  pages = {151--177},
  number = {1-2}
}

@article{hall1935,
  author = {P. Hall},
  title = {{On Representatives of Subsets}},
  year = {1935},
  doi = {10.1112/jlms/s1-10.37.26},
  journal = {Journal of the London Mathematical Society},
  volume = {s1-10},
  pages = {26--30},
  number = {1}
}

@article{liu1973,
  author = {C. L. Liu and James W. Layland},
  title = {{Scheduling Algorithms for Multiprogramming in a Hard-Real-Time Environment}},
  year = {1973},
  doi = {10.1145/321738.321743},
  journal = {Journal of the ACM},
  volume = {20},
  pages = {46--61},
  number = {1}
}

@article{elliott1963,
  author = {E. O. Elliott},
  title = {{Estimates of Error Rates for Codes on Burst-Noise Channels}},
  year = {1963},
  doi = {10.1002/j.1538-7305.1963.tb00955.x},
  journal = {Bell System Technical Journal},
  volume = {42},
  pages = {1977--1997},
  number = {5}
}

@article{kim2014,
  author = {Jin-Woo Kim and Kyeong Hur and Seong-Ro Lee},
  title = {{A Multi-channel Scheduling Scheme for Collision-Free High-Rate WPANs}},
  year = {2014},
  doi = {10.1007/s11277-014-1760-9},
  journal = {Wireless Personal Communications},
  volume = {78},
  pages = {429--448},
  number = {1}
}

@article{tiemann2022,
  author = {Janis Tiemann and Johannes Friedrich and Christian Wietfeld},
  title = {{Experimental Evaluation of IEEE 802.15.4z UWB Ranging Performance under Interference}},
  year = {2022},
  doi = {10.3390/s22041643},
  journal = {Sensors},
  volume = {22},
  pages = {1643},
  number = {4}
}

@misc{zhang2025,
  author = {Fengyun Zhang},
  title = {{A Solution of Ultra Wideband Based High-resolution and Lossless Audio Transmission}},
  year = {2025},
  howpublished = {arXiv:2508.20782},
  note = {Preprint; version explicitly identified},
  url = {https://arxiv.org/abs/2508.20782}
}

@misc{nelson2024,
  author = {Brian Nelson and Hussein Moradi and Behrouz Farhang-Boroujeny},
  title = {{UWB Narrowband Interference Survey and Design Considerations}},
  year = {2024},
  howpublished = {arXiv:2411.07052v1},
  note = {Preprint; version explicitly identified},
  url = {https://arxiv.org/abs/2411.07052v1}
}

@misc{pang2024,
  author = {Bozheng Pang and Tim Claeys and Hans Hallez and Jeroen Boydens},
  title = {{Modeling the Trade-off between Throughput and Reliability in a Bluetooth Low Energy Connection}},
  year = {2024},
  howpublished = {arXiv:2405.01231},
  note = {Preprint; version explicitly identified},
  url = {https://arxiv.org/abs/2405.01231}
}

@article{gilbert1960,
  author = {E. N. Gilbert},
  title = {{Capacity of a Burst-Noise Channel}},
  journal = {Bell System Technical Journal},
  year = {1960},
  volume = {39},
  number = {5},
  pages = {1253--1265},
  doi = {10.1002/j.1538-7305.1960.tb03959.x}
}

@misc{btll,
  author = {{Bluetooth SIG}},
  title = {{Bluetooth Core Specification 5.4, Vol. 6, Part B: Link Layer Specification}},
  howpublished = {[Online]. Available: \url{https://www.bluetooth.com/wp-content/uploads/Files/Specification/HTML/Core-54/out/en/low-energy-controller/link-layer-specification.html}},
  note = {Accessed 2026-10-17}
}

@misc{isoal,
  author = {{Bluetooth SIG}},
  title = {{Bluetooth Core Specification 6.0, Vol. 6, Part G: Isochronous Adaptation Layer}},
  howpublished = {[Online]. Available: \url{https://www.bluetooth.com/wp-content/uploads/Files/Specification/HTML/Core-60/out/en/low-energy-controller/isochronous-adaptation-layer.html}},
  note = {Accessed 2026-10-17}
}

@misc{bap,
  author = {{Bluetooth SIG}},
  title = {{Basic Audio Profile, Version 1.0.1}},
  howpublished = {[Online]. Available: \url{https://www.bluetooth.com/wp-content/uploads/Files/Specification/HTML/16212-BAP-html5/out/en/index-en.html}},
  note = {Accessed 2026-10-17}
}

@misc{ieee4ab,
  author = {{IEEE Standards Association}},
  title = {{IEEE P802.15.4ab: Project Scope and Status}},
  howpublished = {[Online]. Available: \url{https://standards.ieee.org/ieee/802.15.4ab/10694/}},
  note = {Accessed 2026-10-17}
}

@misc{spark,
  author = {{SPARK Microsystems}},
  title = {{SPARK SDK Documentation, Version 2.4.0}},
  howpublished = {[Online]. Available: \url{https://www.sparkmicro.com/sdk-docs/index.html}},
  note = {Accessed 2026-10-17}
}

@misc{lc3white,
  author = {{Bluetooth SIG}},
  title = {{Performance Characterization of the Low Complexity Communication Codec (LC3)}},
  howpublished = {[Online]. Available: \url{https://www.bluetooth.com/wp-content/uploads/2023/07/LC3Characterization_WP.pdf}},
  note = {Accessed 2026-10-17}
}

@misc{ieee4z,
  author = {{IEEE}},
  title = {{IEEE Std 802.15.4z-2020: Enhanced Ultra Wideband Physical Layers (PHYs) and Associated Ranging Techniques}},
  howpublished = {[Online]. Available: \url{https://doi.org/10.1109/IEEESTD.2020.9179124}},
  note = {Accessed 2026-10-17}
}
\end{document}